\documentclass[letterpaper, 10 pt, conference]{ieeeconf}
\IEEEoverridecommandlockouts
\usepackage{inputenc}
\usepackage[american]{babel}
\usepackage{csquotes}
\usepackage[letterpaper,left=48pt,bottom=43pt,right=48pt,top=57pt]{geometry}

\newcommand{\D}{\mathcal{D}}

\newcommand{\R}{\mathbb{R}}

\usepackage{url}
\usepackage{hyperref}
\hypersetup{colorlinks=false,linkcolor=blue,urlcolor=blue}
\usepackage{graphicx}
\usepackage{subcaption}

\addto\captionsenglish{}
\addto\captionsamerican{}
\usepackage{multicol}
\usepackage{xspace}
\usepackage{mathtools}
\usepackage{booktabs}
\usepackage{multirow}
\usepackage{tabularx}
\newcolumntype{C}{>{\centering\arraybackslash}X}
\newcolumntype{L}{>{\raggedright\arraybackslash}X}
\usepackage{siunitx}
\usepackage{listings}
\usepackage{float}
\usepackage{bm}         
\usepackage{amsmath,amssymb,amsfonts}

\newtheorem{theorem}{Theorem}

\newtheorem{proposition}{Proposition}
\newtheorem{remark}{Remark}
\usepackage{bm}

\allowdisplaybreaks
\begin{document}
\title{Virtual Constraint for a Quadrotor UAV Enforcing a Body-Axis Pointing Direction}

\author{Alexandre Anahory Simoes,  Leonardo J. Colombo, Juan I. Giribet, Efstratios Stratoglou%
\thanks{A. Anahory Simoes is with the School of Science and Technology, IE University, Spain
        {\tt\small alexandre.anahory@ie.edu}}
\thanks{L. Colombo is with Centre for Automation and Robotics (CSIC-UPM), Ctra. M300 Campo Real, Km 0,200, Arganda
del Rey - 28500 Madrid, Spain.{\tt\small leonardo.colombo@csic.es}}%
\thanks{Juan I. Giribet is with Universidad de San Andr\'es (UdeSA) and CONICET, Argentina.
        {\tt\small jgiribet@conicet.gov.ar}}\thanks{E. Stratoglou is with American College of Thessaloniki (ACT), 17 Sevenidi St.
55535, Pylaia, Thessaloniki, Greece. {\tt\small stratogl@act.edu}}
\thanks{L. Colombo acknowledge financial support from Grant PID2022-137909NB-C21 funded by MCIN/AEI/ 10.13039/501100011033. The research leading to these results was supported in part by iRoboCity2030-CM, Robótica Inteligente para Ciudades Sostenibles (TEC-2024/TEC-62), funded by the Programas de Actividades I+D en Tecnologías en la Comunidad de Madrid. J. Giribet was supported by PICT-2019-2371 and PICT-2019-0373 projects from Agencia Nacional de Investigaciones Cient\'ificas y Tecnol\'ogicas, and UBACyT-0421BA project from the Universidad de Buenos Aires (UBA), Argentina.}%
}

\maketitle
\begin{abstract}
We propose a geometric control framework on $SE(3)$ for quadrotors that enforces pointing-driven missions without completing a full attitude reference. The mission is encoded through virtual constraints defining a task manifold and an associated set of admissible velocities, and invariance is achieved by a feedback law obtained from a linear system in selected inputs. Under a transversality condition with the effective actuation distribution, the invariance-enforcing input is uniquely defined, yielding a constructive control law and, for relevant tasks, closed-form expressions. We further derive a local off-manifold stabilization extension. As a case study, we lock a body axis to a prescribed line-of-sight direction while maintaining fixed altitude.
\end{abstract}



\section{Introduction}\label{sec:intro}
Quadrotor UAVs are increasingly deployed in missions where the \emph{direction} of an onboard body-fixed axis matters as much as, or even more than, classical position and attitude tracking. Typical examples include camera pointing and vision-based servoing for target-relative maneuvers and line-of-sight objectives~\cite{Thomas2016RAL,Sheckells2016IROS,HamelMahony2002TRO}, aerial filming where the drone/camera motion must respect pointing and visibility constraints~\cite{Sabetghadam2019ECMR}, and inspection/mapping scenarios where directional sensing and persistent visibility are central~\cite{9636824,yang2024safety}. In these settings, the controller must shape the motion so that a chosen body axis remains aligned with a desired direction while satisfying additional geometric requirements such as fixed altitude or planar flight.

While gimbals can decouple camera pointing from vehicle attitude in some platforms, they are not always available, may have limited range/bandwidth, and do not address missions where the vehicle body axis itself must satisfy geometric constraints (e.g., directional sensing, aerodynamic/propulsion-aligned instruments, or safety/visibility constraints coupled to the thrust direction). Our approach targets this latter class by enforcing mission geometry at the UAV level, without relying on full-attitude reference completion

Most existing approaches treat such objectives through trajectory tracking: one constructs a feasible reference $(p_d(t),R_d(t))\in\mathbb{R}^{3}\times SO(3)$---typically by prescribing a desired thrust direction and a yaw profile---and then applies a geometric tracking controller on $SE(3)$~\cite{Lee2010CDC,MellingerKumar2011ICRA}. 
This is effective when a consistent full attitude reference is naturally available, but it becomes less direct when the task constrains only \emph{part} of the attitude (e.g., the thrust direction or a single body axis) and must be met together with other geometric constraints. 
In practice, the remaining degrees of freedom are completed via heuristics or cascaded architectures (outer-loop position, inner-loop attitude/rate), or embedded in online planning/optimization and allocation layers~\cite{Bangura2017Thesis,RichterBryRoy2013ISRR,BryRichterBachrachRoy2015IJRR}. 
Moreover, since pointing tasks couple translation and rotation through the underactuated thrust direction, reference completion can obscure basic feasibility and well-posedness questions; for instance, common attitude parametrizations unnecessarily couple control of the thrust direction with yaw tracking, which can degrade position/pointing performance~\cite{kooijman2019trajectory}.

Even within a tracking paradigm, completing a full attitude reference can introduce undesirable couplings: many designs regulate the thrust direction on $\mathbb{S}^2$ together with yaw on $\mathbb{S}^1$, so yaw errors can leak into the thrust direction and degrade translational/pointing performance. This has motivated ``reduced-attitude'' controllers on $\mathbb{S}^2\times \mathbb{S}^1$ that explicitly decouple thrust-direction control from yaw~\cite{kooijman2019trajectory}. Here we take a different route and avoid reference completion altogether by encoding the mission as an invariant geometric object and enforcing it by construction.

We do so using \emph{virtual nonholonomic constraints} \cite{Simoes:linear:nonholonomic, stratoglou2024geometry, stratoglou2024virtualc, simoes2025geometric}. Rather than tracking a completed reference $(p_d(t),R_d(t))$, we define a \emph{task manifold} in $SE(3)$ (configurations that satisfy the task constraints) and an associated \emph{task distribution} in $TSE(3)$ (velocities that preserve them), specified via linear velocity constraints in body coordinates. Enforcing invariance amounts to differentiating these constraints along the dynamics and requiring that they remain satisfied, which yields a linear system in a selected set of control inputs. If the task distribution is transversal to the corresponding effective actuation distribution, this system is nonsingular and the invariance-enforcing input is uniquely defined, leading to a constructive control law and, in relevant cases, closed-form thrust and torque expressions.

More broadly, the proposed virtual-constraint viewpoint is most relevant when the mission specifies only a subset of the configuration/attitude variables and the remaining degrees of freedom are not directly actuated or are intentionally left free (e.g., restricted torque channels, thrust-direction coupling, or limited authority). In such cases, enforcing an invariant task geometry provides a principled alternative to completing and tracking a full reference.

Our objective is to provide a transparent alternative to reference completion for pointing-driven missions, together with a local stabilization mechanism for velocity-level task compatibility. The proposed control law makes feasibility and well-posedness explicit: whenever transversality holds, the thrust and torque components required to maintain the task are obtained from a uniquely solvable linear system. {In addition, we derive an off-manifold stabilization extension that acts directly on the velocity-level constraint residuals and recovers the invariance-enforcing law on the task distribution.} As a case study, we obtain closed-form expressions that lock a body axis to a prescribed line-of-sight direction while maintaining fixed altitude. This shifts the design focus from tracking an arbitrarily completed attitude trajectory to maintaining (and locally recovering) mission-compatible geometry, which is natural in directional sensing and line-of-sight applications where only a subset of attitude is operationally meaningful.

In particular, the main contributions of this work are: (i) We formulate pointing-driven quadrotor missions on $SE(3)$ via virtual nonholonomic constraints, by defining a task manifold in configuration space together with an associated task distribution of admissible velocities. (ii) We derive a constructive invariance-enforcing controller: the required inputs are obtained from a linear system induced by differentiating the velocity constraints, and are uniquely determined under a transversality condition with the effective actuation distribution. (iii) We derive a local off-manifold stabilization extension for the velocity-level compatibility conditions, with explicit feedback laws that yield exponential residual decay on a feasible set and recover the invariance-enforcing law on the constraint set. (iv) We apply the framework in a body-axis line-of-sight constraint  on a quadrotor combined with fixed altitude, yielding closed-form thrust and torque expressions.

The paper is structured as follows. Section~\ref{sec:virtual_constraints} introduces the virtual-constraint formulation and the invariance-enforcing control on $SE(3)$. Section~\ref{sec:application} applies the method to body-axis pointing with fixed altitude and derives explicit control expressions. {In addition, we develop a  constraint stabilization for the pointing--altitude task.} Section~\ref{sec:simulations} presents numerical simulations. Section~\ref{sec:conclusions} concludes and outlines future directions.

\section{Virtual constraints in $SE(3)$}\label{sec:virtual_constraints}

Let $SE(3)=SO(3)\times \mathbb{R}^3$ with elements $g=(R,p)$, $R\in SO(3)$, $p\in\mathbb{R}^3$.
We identify the Lie algebra $\mathfrak{se}(3)=\mathfrak{so}(3)\times \mathbb{R}^3$ and use the hat map
$\widehat{\cdot}:\mathbb{R}^3\to\mathfrak{so}(3)$ so that $\widehat{x}y=x\times y$.
A tangent vector $(\dot R,\dot p)\in T_{(R,p)}SE(3)$ admits a \emph{left-trivialized} representation
\begin{equation}\label{eq:left-trivialized-velocity}
(\dot R,\dot p) = (R\widehat{\Omega},\, v),\qquad (\Omega,v)\in\mathbb{R}^3\times\mathbb{R}^3,
\end{equation}
and we write the corresponding left-trivialized body velocity as
\[
\xi := (\Omega,v)\in\mathfrak{se}(3)\simeq \mathbb{R}^6.
\]

For convenience we denote the body axes by
\begin{equation}\label{eq:body-axes}
b_i := R e_i,\qquad i=1,2,3,
\end{equation}
and note that $\dot b_i=b_i\times \Omega$ whenever $\dot R=R\widehat{\Omega}$.

Consider the mechanical system whose configuration space is $SE(3)=SO(3)\times \R^{3}$ and models the dynamics of an underactuated aerial robot of quadrotor type, with control inputs $u = (f,\tau) \in \mathbb{R} \times \mathbb{R}^3$, where $f$ is the total thrust generated along the body-fixed $e_3$ axis and $\tau \in \mathbb{R}^3$ are the control torques. The equations of motion are given by
\begin{align}
\dot{p} &= v, \label{eq:transl_kin}\\
m \dot{v} &= f R e_3 - m g e_3, \label{eq:transl_dyn}\\
\dot{R} &= R \hat{\Omega}, \label{eq:rot_kin}\\
J \dot{\Omega} + \Omega \times J \Omega &= \tau,
\label{eq:rot_dyn}
\end{align}
where $m>0$ is the mass, $J \in \mathbb{R}^{3\times 3}$ is the inertia
matrix, $g>0$ is the gravitational constant. We assume the UAV rigid-body attitude dynamics is modeled with a diagonal inertia matrix
\begin{equation}\label{eq:J-diag}
J=\mathrm{diag}(J_1,J_2,J_3).
\end{equation}

A \textit{virtual constraint} is a submanifold $\D\subseteq TSE(3)$ that is control invariant for the closed loop system \eqref{eq:transl_kin}, \eqref{eq:transl_dyn}, \eqref{eq:rot_kin} and \eqref{eq:rot_dyn} for a suitable choice of control law. In other words, there exists an admissible feedback control law that makes $\mathcal{D}$ an invariant submanifold for the closed-loop system.

It has already been established in earlier works (see, for example \cite{Simoes:linear:nonholonomic, stratoglou2024geometry}) that if $\D$ is transversal to the actuated directions (which are at most $4$, implying that $\D$ must have dimension $6+k$ with $2 \leqslant k \leqslant 5$) then there will be a unique control law enforcing the virtual constraint. In our main theoretical result, we will derive an expression for the control law achieving our desired task.

\subsection{Constraint distribution}
Fix an integer $m\in\{1,2,3,4\}$. We consider $m$ (smooth) \emph{constraint functions}
\begin{equation*}
\mu^a :T(SE(3))\to \R,\qquad a=1,\dots,m,
\end{equation*}
which are linear on the velocities and define the associated \emph{constraint (nonholonomic) distribution}
\begin{equation*}
\D_{g} := \{(g,\Omega, v)\in TSE(3) \ | \ \mu^{a}(g,\Omega, v)=0, \ a=1,\ldots, m\}.
\end{equation*}
At each $g\in SE(3)$, we will assume that the fiber $\D_g\subset T_g SE(3)$ is a linear subspace of codimension $m$. The functions $\mu^{a}$ are $m$ linearly independent one-forms. Hence $\D$ has constant rank $\dim \D_{g} = 6-m$ on $\mathcal{U}$.

Using the left trivialization \eqref{eq:left-trivialized-velocity}, each constraint one-form can be written as
\begin{equation}\label{eq:alpha-body-form}
\mu^a_{(R,p)}(\dot R,\dot p)= C^{a}(R,p) \cdot \xi = A^a(R,p)\cdot \Omega + B^a(R,p)\cdot v,
\end{equation}
for uniquely defined coefficient maps $A^a,B^a:SE(3)\to\mathbb{R}^3$.

The problem in the next section addresses the special case where the distribution $\D$ is \emph{integrable}. A simple test to prove integrability is to check if it is closed under Lie brackets, i.e., for all smooth vector fields $X,Y$ with $X(g),Y(g)\in\D_g$ one has $[X,Y](g)\in\D_g$.
When $\D$ is integrable, through each point $g\in\mathcal{U}$ there exists a maximal immersed submanifold
$\mathcal{M}\subset SE(3)$ whose tangent bundle satisfies
\[
T_g\mathcal{M}=\D_g,\qquad \forall g\in\mathcal{M},
\]
and we may interpret $\mathcal{M}$ as the configuration constraint manifold. Intuitively, the constraint on velocities are derived from constraints on positions.

A sufficient condition for integrability in the present setting is that the $\mu^a$ arise as differentials of $m$ independent smooth functions
\begin{equation}\label{eq:holonomic-constraints}
\phi^a:SE(3)\to\mathbb{R},\qquad \mu^a = d\phi^a,\qquad a=1,\dots,m,
\end{equation}
with $d\phi^{a}$ linearly independent. Then
\begin{equation}\label{eq:submanifold-levelset}
\mathcal{M}:=\{g\in SE(3):\ \phi^a(g)=0,\ a=1,\dots,m\}
\end{equation}
is a submanifold of codimension $m$ and $T\mathcal{M}=\D$.

\begin{remark}
    Previous papers introduced the notion of virtual holonomic constraints (see \cite{maggiore2012virtual} for instance), which are virtual constraints arising from restricting the admissible configurations only and not the velocities. However, to make sure that a dynamical system does not violate a given set of admissible configurations, one has also to enforce that its velocity is compatible with the admissible configurations. This set of admissible velocities forms in turn a distribution $\D$, like the one defined above. So, in some sense, there is no essential difference between the concept of virtual holonomic constraint or virtual nonholonomic constraint, under the understanding that a restriction on velocities may also arise from a restriction on the configurations, which corresponds to the case when $\D$ is integrable. 
\end{remark}

\subsection{The control distributions}
\label{subsec:affine_xi_level}

The UAV dynamics is \emph{second order} on $SE(3)$, the controls $(f,\tau)$ do not enter affinely at the configuration level $\dot g\in T_gSE(3)$, but rather in the dynamics of the velocities.


Projecting equations \eqref{eq:transl_kin}, \eqref{eq:transl_dyn}, \eqref{eq:rot_kin} and \eqref{eq:rot_dyn} to the velocity space $\xi=(\Omega, v)\in \R^{6}$ we obtain
\begin{equation}\label{eq:xi_affine}
\dot\xi = \underbrace{\begin{bmatrix} -J^{-1}(\Omega\times J\Omega)  \\[0.5mm]-g e_3\end{bmatrix}}_{=:a_0(\xi)}
\;+\;
f\underbrace{\begin{bmatrix}0  \\[0.5mm] \frac{1}{m}R e_3\end{bmatrix}}_{=:a_f(g)}
\;+\;
\sum_{i=1}^3 \tau_i\underbrace{\begin{bmatrix} J^{-1}e_i \\[0.5mm]0\end{bmatrix}}_{=:a_{\tau_i}},
\end{equation}
where we emphasize that the thrust direction $a_f$ depends on $R$ (hence on $g$), while the torque directions
$a_{\tau_i}$ are constant in these coordinates.

Define the \emph{total control distribution} as the tangent subspace spanned by the actuated directions
\begin{equation}\label{eq:F_xi_def}
\mathcal{F}_{g} := \mathrm{span}\Big\{a_f(g),\, a_{\tau_1},\,a_{\tau_2},\,a_{\tau_3}\Big\}
\ \subset\ \mathbb{R}^6.
\end{equation}
If, for example $\tau_1\equiv 0$, we can also define an \textit{effective control distribution} given by
\begin{equation}\label{eq:F_xi_tau1zero}
\mathcal{F}_{g}^{\mathcal{I}}=\mathrm{span}\big\{a_f(g),\,a_{\tau_2},\,a_{\tau_3}\big\}.
\end{equation}



\subsection{Invariant controller}
\label{controller:section}

The following proposition is an immediate consequence of a more general theorem proved in \cite{Simoes:linear:nonholonomic}. We will say that two distributions $\mathcal{A}_{1}$ and $\mathcal{A}_{2}$ on a manifold $Q$ are said to be transversal if they are complementary, in the sense that $TQ=\mathcal{A}_{1}\oplus\mathcal{A}_{2}$.

\begin{proposition}\label{existence:result}
If the distribution $\mathcal{D}$ and the effective control distribution $\mathcal{F}^{\mathcal{I}}$ are transversal, then there exists a unique control function making the distribution a virtual constraint associated with the mechanical control system defined by equations \eqref{eq:transl_kin}, \eqref{eq:transl_dyn}, \eqref{eq:rot_kin} and \eqref{eq:rot_dyn}.
\end{proposition}

The transversality assumption is a geometric condition that is equivalent to a certain regularity assumption on the constraint equations, allowing to solve them for a unique control law ensuring their invariance.

We will derive next a formula for the unique control law under the additional hypothesis that the thrust direction $a_{f}$ belongs to the effective control distribution. Let $m\in\{1,2,3, 4\}$ denote the total number of control inputs. In the UAV model these inputs are the scalar thrust $f$ and $m-1$ selected body torques $\tau_i$ with indices $\mathcal{I}=\{i_1,\dots,i_{m-1}\}\subseteq\{1,2,3\}$,  $u:=\big(f,\tau_{i_1},\ldots,\tau_{i_{m-1}}\big)^\top\in\mathbb{R}^m$.

Let $\mathcal{D}\subset TSE(3)$ be a smooth constant-rank distribution of rank $6-m$ and transversal to the effective control distribution $\mathcal{F}_{g}^{\mathcal{I}}$. The requirement that $\dot g\in\mathcal{D}_g$ is the \emph{velocity-level} condition
\begin{equation}\label{eq:vel_level_condition}
\mu^a_g(\dot g)=0,\qquad a=1,\dots,m,
\end{equation}
which, after substituting $\dot g=(R\widehat{\Omega},v)$, imposes $m$ algebraic relations among $(g,\xi)$ of the form \eqref{eq:alpha-body-form}.

To guarantee \emph{invariance} of the constraint (i.e.\ tangency of the closed-loop vector field to the constraint set),
one must enforce
\begin{equation}\label{eq:inv_requirement}
\frac{d}{dt}\Big(C^{a}(R,p) \cdot \xi\Big)=0,\qquad a=1,\dots,m,
\end{equation}
along solutions that satisfy \eqref{eq:vel_level_condition}. This is equivalent to
\begin{equation}
\Big(dC^{a}(R,p) \cdot (\dot{R}, \dot{p})\Big) \cdot \xi + C^{a}(R,p) \cdot \dot{\xi}=0,
\end{equation}
and using \eqref{eq:xi_affine} we obtain
\begin{equation*}
\Big(dC^{a}(R,p) \cdot (\dot{R}, \dot{p})\Big) \cdot \xi + C^{a}(R,p) \cdot \Big( a_{0} + f a_{f} + \tau_{i} a_{\tau_{i}}\Big)=0.
\end{equation*}
Introduce the $m$ functions $G^{a}$ absorbing all the terms that do not contain the controls
$$G^{a} = \Big(dC^{a}(R,p) \cdot (\dot{R}, \dot{p})\Big) \cdot \xi + C^{a}(R,p) \cdot a_{0}.$$
Then we obtain the linear system of equations with $m$ equations and $m$ unknowns
$$G^{a} +  f C^{a}(R,p)\cdot a_{f} + \tau_{i} C^{a}(R,p)\cdot a_{\tau_{i}}=0,$$
which has a unique solution under the transversality assumption and can be obtained by inverting the matrix whose columns are $C^{a}(R,p)\cdot a_{f}$, $C^{a}(R,p)\cdot a_{\tau_{i}}$.

%



These controllers are unique by Proposition \ref{existence:result}.

\section{Virtual Constraint Enforcing a Body-Axis Pointing Direction}\label{sec:application}

Consider the quadrotor UAV model presented in the previous section. The description of the configuration is given by $(R,p)\in SE(3)$ with $R\in SO(3)$ and $p\in\mathbb{R}^3$, and the body-frame axes were $b_i := R e_i,\, i=1,2,3,$ and let $\rho := \|p\|>0$, and $\hat{p} := \frac{p}{\rho}$. We first consider the \textit{feasible submanifold}
$$\mathcal{M}_{0} := \{(R,p)\in SE(3):\ e_3^\top R e_3 \neq 0, \quad p\neq 0\},$$
where our controller will be well-defined. Next, we consider the constraint submanifold
\begin{equation}\label{eq:M-def}
\mathcal{M} := \{(R,p)\in \mathcal{M}_{0}:\ b_1 = -\hat{p}\},
\end{equation}
describing a locked direction and add a fixed-height constraint $e_3^\top p=z_0$ to define the submanifold
\begin{equation}\label{eq:M1-def}
\mathcal{M}_1 := \{(R,p)\in \mathcal{M}:\ b_1 = -\hat{p},\ \ e_3^\top p=z_0\}.
\end{equation}
We assume the UAV dynamics on $TSE(3)$ are given (as standard) and we restrict attention to the controlled directions produced by thrust $f$ and torque $\tau\in\mathbb{R}^3$, with the additional assumption $\tau_1=0$.

\subsection{Constraints in body coordinates}

On $\mathcal{M}$ the attitude constraint $b_1=-\hat{p}$ says the first body axis points along the \emph{negative} radial direction. Differentiating
\[
b_1 + \hat{p} = 0
\]
yields the velocity-level compatibility condition
\begin{equation}\label{eq:vel-level}
\dot b_1 + \dot{\hat{p}} = 0.
\end{equation}
Using $\dot b_1 = (R \Omega)\times b_1$ (since $\dot R = R\widehat\Omega$) and the standard identity
\[
\dot{\hat{p}}  = \frac{1}{\rho}(I-\hat{p}\hat{p}^\top)\dot p = \frac{1}{\rho}(I-\hat{p}\hat{p}^\top)v,
\]
equation \eqref{eq:vel-level} becomes
\begin{equation}\label{eq:vel-level-expanded}
(R \Omega)\times b_1 + \frac{1}{\rho}(I-\hat{p}\hat{p}^\top)v=0.
\end{equation}
On $\mathcal{M}$ we have $\hat{p}=-b_1$, hence $I-\hat{p}\hat{p}^\top = I-b_1b_1^\top$ is the orthogonal projector onto $\mathrm{span}\{b_2,b_3\}$. Writing
\[
v=v_1 b_1+v_2 b_2+v_3 b_3,\qquad v_i:=b_i^\top v,
\]
and $R\Omega=\Omega_1 b_1+\Omega_2 b_2+\Omega_3 b_3$ (equivalently $\Omega_i=e_i^\top \Omega$), we compute
\[
(R \Omega)\times b_1 = \Omega_3 b_2 - \Omega_2 b_3,\qquad (I-b_1b_1^\top)v=v_2 b_2+v_3 b_3.
\]
Therefore \eqref{eq:vel-level-expanded} is equivalent to the two scalar constraints
\begin{equation}\label{eq:Omega23-constraints}
\Omega_3 = \frac{v_2}{\rho},\qquad \Omega_2 = -\frac{v_3}{\rho}.
\end{equation}

The height constraint $e_3^\top p=z_0$ gives the additional velocity-level condition
\begin{equation}\label{eq:height-vel}
e_3^\top v=0.
\end{equation}

Thus, equations \eqref{eq:alpha-body-form} are three constraint equations taking the form
\begin{equation*}
    \begin{split}
        \mu^{1} & = e_3^\top v \\
        \mu^{2} & = e_3^\top \Omega - \frac{1}{\rho}b_{2}^\top v \\
        \mu^{3} & = e_2^\top \Omega + \frac{1}{\rho}b_{3}^\top v.
    \end{split}
\end{equation*}
Thus the three effective control directions are enough to ensure a unique control law enforcing the desired constraint, provided that they fulfill the transversality assumption.

\subsection{Transversality}

The tangent space at $(R,p)\in\mathcal{M}_1$ is the set of vectors $(\dot{R},\dot{p})=(R\widehat\Omega,\,v)$ satisfying the linearized constraints \eqref{eq:Omega23-constraints} and \eqref{eq:height-vel}. Thus $T_{(R,p)}\mathcal{M}_1$ is a codimension-$3$ subspace of $T_{(R,p)}SE(3)$. The control directions on configuration induced by $f$ and $\tau$ are given by the effective control distribution \eqref{eq:F_xi_tau1zero}.

At each $(R,p)\in\mathcal{M}_1$ the three control directions are linearly independent and satisfy
\[
\mathcal{F}_{(R,p)}^{\mathcal I}\ \cap\ T_{(R,p)}\mathcal{M}_1=\{0\},
\]
for $p\neq 0$ and $e_3^\top R e_3 \neq 0$, i.e., if the body axis $b_{3}$ is orthogonal to the inertial axis $e_{3}$, the transversality condition is violated and the virtual constraint is not feasible. Under these conditions, they provide a complement to $T_{(R,p)}\mathcal{M}_1$ inside $T_{(R,p)}SE(3)$.

Indeed, independence is immediate since $a_f$ has nonzero only translational part while $a_{\tau_{2}},a_{\tau_{3}}$ have only rotational parts, and the latter correspond to independent generators of rotations. In addition, none of these vectors belongs to \(T_{(R,p)}M_1\) if $e_{3}^\top b_{3} \neq 0$, that is, provided we remain in the feasibility set $\mathcal{M}_{0}$, the two distributions are transversal. Intuitively, this translates in the fact that we are not able to enforce the chosen constraints when the body axis $b_{3}$, directed along the thrust, is orthogonal to the inertial vertical axis. In case $e_{3}^\top b_{3} = 0$, the direction span by $(\Omega, v)=(e_{2}, b_{3})$ is in the intersection. 

\subsection{Virtual constraint}

We now derive feedback laws for $(f,\tau_2,\tau_3)$ ensuring that $\mathcal{M}_1$ is invariant for the closed-loop dynamics, assuming the system starts on $\mathcal{M}_1$ and satisfies the velocity-level compatibility conditions \eqref{eq:Omega23-constraints}--\eqref{eq:height-vel}.

\begin{theorem}
    Suppose that the control system \eqref{eq:transl_kin}-\eqref{eq:rot_dyn} is actuated by $(f,\tau_2,\tau_3)$. Then, the thrust controller
    $$f=\frac{mg}{e_3^\top R e_3},$$
    and the torque controllers given by
    
\begin{align}
\tau_2 &= J_2\Big(\Omega_1\Omega_3 + \frac{1}{\rho}\Big(-\frac{f}{m}+g s_3\Big)\Big)+(J_1-J_3)\Omega_1\Omega_3,\nonumber\\
\tau_3 &= J_3\Big(-\Omega_1\Omega_2 - \frac{g}{\rho}s_2\Big)+(J_2-J_1)\Omega_1\Omega_2,\label{eq:tau23-laws}
\end{align} are the unique controllers ensuring that the submanifold \eqref{eq:M1-def} is a virtual constraint for this control system, where
    \[
        s_2=e_3^\top b_2=e_3^\top R e_2,\qquad s_3=e_3^\top b_3=e_3^\top R e_3.
    \]
\end{theorem}

\begin{proof}
From $e_3^\top p=z_0$ we require $e_3^\top v=0$ and
\[
0=\frac{d}{dt}(e_3^\top v)=e_3^\top \dot v.
\]
Using the translational dynamics $m\dot v=f b_3-mg e_3$ yields
\[
0=e_3^\top\dot v=\frac{f}{m}e_3^\top b_3-g,
\]
hence, provided $e_3^\top b_3\neq 0$,
\begin{equation}\label{eq:f-law}
{\quad f=\frac{mg}{e_3^\top b_3}=\frac{mg}{e_3^\top R e_3}.\quad}
\end{equation}

To keep $b_1=-\hat{p}$ invariant, it suffices to enforce that the velocity-level constraints \eqref{eq:Omega23-constraints} are preserved by the dynamics. Differentiate \eqref{eq:Omega23-constraints}:
\begin{equation}\label{eq:diff-Omega23}
\dot\Omega_3 = \frac{d}{dt}\Big(\frac{v_2}{\rho}\Big),\qquad
\dot\Omega_2 = -\frac{d}{dt}\Big(\frac{v_3}{\rho}\Big).
\end{equation}
We compute the required right-hand sides.

First, note that $\dot\rho = \hat{p}^\top \dot p = \hat{p}^\top v = -b_1^\top v = -v_1$, hence
\begin{equation}\label{eq:rho-dot}
\dot\rho = -v_1.
\end{equation}
Next, since $v_i=b_i^\top v$,
\[
\dot v_i = \dot b_i^\top v + b_i^\top \dot v,\qquad \dot b_i=(R \Omega)\times b_1.
\]

Using $\dot v=\frac{f}{m}b_3-ge_3$ we obtain:
\begin{align*}
\dot v_2 &= ((R \Omega)\times b_2)^\top v + b_2^\top \Big(\frac{f}{m}b_3-ge_3\Big)\\
&= ((R \Omega)\times b_2)^\top v - g\, s_2,\\
\dot v_3 &= ((R \Omega)\times b_3)^\top v + b_3^\top \Big(\frac{f}{m}b_3-ge_3\Big)\\
&= ((R \Omega)\times b_3)^\top v + \frac{f}{m}- g\, s_3.
\end{align*}
A direct expansion in the body basis gives
\[
((R \Omega)\times b_2)^\top v =- v_1\Omega_3 + v_3\Omega_1,\quad\]
\[
((R \Omega)\times b_3)^\top v = v_1\Omega_2 - v_2\Omega_1.
\]
Substituting the on-manifold relations $v_2=\rho\Omega_3$ and $v_3=-\rho\Omega_2$ from \eqref{eq:Omega23-constraints} yields
\begin{align}
\dot v_2 &= -v_1\Omega_3 - \rho\Omega_2\Omega_1 - g s_2,\label{eq:vdot2-simpl}\\
\dot v_3 &= v_1\Omega_2 - \rho\Omega_3\Omega_1 + \frac{f}{m}- g s_3.\label{eq:vdot3-simpl}
\end{align}
Combining \eqref{eq:diff-Omega23}--\eqref{eq:rho-dot} with \eqref{eq:vdot2-simpl}--\eqref{eq:vdot3-simpl} gives the required accelerations:
\begin{equation}\label{eq:Omega23dot-required}
{\begin{aligned}
\dot\Omega_3 &= -\Omega_1\Omega_2 - \frac{g}{\rho}s_2,\\[1mm]
\dot\Omega_2 &= \Omega_1\Omega_3 + \frac{1}{\rho}\Big(-\frac{f}{m}+g s_3\Big).
\end{aligned}}
\end{equation}

Assume $J=\mathrm{diag}(J_1,J_2,J_3)$. Then
\[
J\dot\Omega+\Omega\times J\Omega=\tau,
\]
and the second and third components satisfy
\begin{align*}
\dot\Omega_2&=\frac{1}{J_2}\Big(\tau_2-(J_1-J_3)\Omega_1\Omega_3\Big),\\
\dot\Omega_3&=\frac{1}{J_3}\Big(\tau_3-(J_2-J_1)\Omega_1\Omega_2\Big).
\end{align*}
Imposing \eqref{eq:Omega23dot-required} gives the feedback laws
\begin{equation}
{\begin{aligned}
\tau_2 &= J_2\Big(\Omega_1\Omega_3 + \frac{1}{\rho}\Big(-\frac{f}{m}+g s_3\Big)\Big)+(J_1-J_3)\Omega_1\Omega_3,\\
\tau_3 &= J_3\Big(-\Omega_1\Omega_2 - \frac{g}{\rho}s_2\Big)+(J_2-J_1)\Omega_1\Omega_2,
\end{aligned}}
\end{equation}
together with $\tau_1=0$.

Finally, substituting the thrust law \eqref{eq:f-law} into \eqref{eq:tau23-laws} yields a closed form in terms of $(R,p,v,\Omega)$ only. 
\end{proof}


\subsection{Constraint stabilization for the pointing--altitude task}
\label{subsec:stabilization_pointing_altitude}

The control law derived in the previous subsection enforces exact invariance of the task distribution when the initial condition is compatible with the velocity constraints. We now derive an off-manifold stabilization extension for the same quadrotor task, inspired by the general virtual-constraint stabilization viewpoint in~\cite{simoes2025geometric, simoes2025geometric-b}. The goal is to drive the velocity-level constraint residuals to zero exponentially while recovering the invariance-enforcing law on the constraint set.

Rewrite the constraint equations as
\begin{equation}
\mu^1 := e_3^\top v,
\qquad
\mu^2 := \Omega_3 - \frac{v_2}{\rho},
\qquad
\mu^3 := \Omega_2 + \frac{v_3}{\rho}.
\label{eq:mu_residuals}
\end{equation}
$\mu^i$ are not only equations defining the constraints but also \textit{residuals} because they measure the instantaneous violation of the velocity-level compatibility equations defining the task distribution.



Differentiating \eqref{eq:mu_residuals} and using \eqref{eq:vdot2-simpl}, \eqref{eq:vdot3-simpl} and Euler's equations \eqref{eq:rot_dyn} yields
\begin{align}
\dot \mu^1
&= e_3^\top \dot v
= \frac{f}{m}\,e_3^\top b_3 - g .
\label{eq:mu3dot}
\\[0.5em]
\dot \mu^2
&= \dot \Omega_3 - \frac{\dot v_2}{\rho} + \frac{v_2\dot \rho}{\rho^2} \nonumber\\
&= \frac{J_1-J_2}{J_3}\,\Omega_1\Omega_2 + \frac{\tau_3}{J_3}
- \frac{\Omega_1 v_3 - \Omega_3 v_1 - g\,s_2}{\rho}\nonumber
\\&+ \frac{v_2\dot \rho}{\rho^2},
\label{eq:mu1dot}
\\[0.5em]
\dot \mu^3
&= \dot \Omega_2 + \frac{\dot v_3}{\rho} - \frac{v_3\dot \rho}{\rho^2} \nonumber\\
&= \frac{J_3-J_1}{J_2}\,\Omega_3\Omega_1 + \frac{\tau_2}{J_2}
+ \frac{\Omega_2 v_1 - \Omega_1 v_2 + \frac{f}{m} - g\,s_3}{\rho}\nonumber
\\&- \frac{v_3\dot \rho}{\rho^2},
\label{eq:mu2dot}
\end{align}

Choose gains $k_1,k_2,k_3>0$ and impose the first-order target dynamics
\begin{equation}
\dot \mu_i = -k_i \mu_i, \qquad i=1,2,3.
\label{eq:target_mu_dynamics}
\end{equation}
Substituting \eqref{eq:mu3dot}--\eqref{eq:mu2dot} into \eqref{eq:target_mu_dynamics} yields a linear system in the selected inputs $(f,\tau_2,\tau_3)$, which can be solved explicitly.

First, from \eqref{eq:mu3dot}:
\begin{equation}
f
=
m\,\frac{g-k_3\mu_3}{e_3^\top b_3},
\qquad e_3^\top b_3\neq 0.
\label{eq:f_stab}
\end{equation}

Next, from \eqref{eq:mu2dot} and using $\dot{\rho}=\hat{p}^{\top}v$:
\begin{multline}
    \tau_2
=
J_2\!\left[
-k_3\mu_3
-\frac{J_3-J_1}{J_2}\Omega_3\Omega_1 \right. \\
\left. -\frac{\Omega_2 v_1 - \Omega_1 v_2 + \frac{f}{m} - g\,s_3}{\rho}
+\frac{v_3 \; \hat{p}^{\top}v}{\rho^2}
\right].
\label{eq:tau2_stab}
\end{multline}


Finally, from \eqref{eq:mu1dot}:
\begin{multline}
\tau_3
=
J_3\!\left[
-k_2\mu_2
-\frac{J_1-J_2}{J_3}\Omega_1\Omega_2 \right.
\\ \left. +\frac{\Omega_1 v_3 - \Omega_3 v_1 - g\, s_2}{\rho}
-\frac{v_2 \; \hat{p}^{\top}v}{\rho^2}
\right].
\label{eq:tau3_stab}
\end{multline}


Equations \eqref{eq:f_stab}--\eqref{eq:tau3_stab} define a feedback law that drives \(\mu_i\) exponentially to zero on the feasible set $\mathcal{M}_{0}$ where the denominators are well defined. By construction, when \(\mu=0\) the correction terms \(-k_i\mu_i\) vanish and the law reduces to the invariance-enforcing controller obtained in the previous subsection.

Note that on the feasible submanifold $\mathcal{M}_{0}$, by construction, the controllers \eqref{eq:f_stab}--\eqref{eq:tau3_stab} are well-defined and the factors \(\rho^{-1}\), \(\rho^{-2}\), and \((e_3^\top b_3)^{-1}\) are non-vanishing.


\begin{proposition}
Assume that the closed-loop trajectory remains in the feasible submanifold \(\mathcal{M}_{0}\) for \(t\in[0,T]\). Under the feedback law \eqref{eq:f_stab}--\eqref{eq:tau3_stab}, the residuals \eqref{eq:mu_residuals} satisfy
\[
\dot\mu_i = -k_i\mu_i,\qquad i=1,2,3,
\]
on \([0,T]\), and therefore
\[
\mu_i(t)=\mu_i(0)e^{-k_i t}, \qquad i=1,2,3.
\]
In particular, the velocity-level task constraints are exponentially recovered on any interval for which the closed-loop solution remains in \(\mathcal{M}_{0}\).
\end{proposition}

\begin{proof}
Substituting \eqref{eq:f_stab}, \eqref{eq:tau2_stab}, and \eqref{eq:tau3_stab} into \eqref{eq:mu1dot}--\eqref{eq:mu3dot} gives \eqref{eq:target_mu_dynamics} identically. The explicit solution follows immediately. 
\end{proof}

\begin{remark}
The stabilization above acts on the \emph{velocity-level} compatibility conditions associated with the pointing and altitude task. Hence it guarantees exponential convergence of the residuals \eqref{eq:mu_residuals} while the trajectory remains in the feasible manifold $\mathcal{M}_{0}$. Convergence to the full configuration-level task manifold may additionally require task-specific geometric conditions and suitable initialization of the holonomic part.
\end{remark}

\section{Simulation Results}\label{sec:simulations}

In this section, we present numerical simulations illustrating the closed-loop behavior of the quadrotor under the proposed virtual nonholonomic constraints for body-axis pointing and fixed altitude. The simulations are designed to assess two aspects: (i) preservation of the task geometry under the invariance-enforcing law (and its practical behavior under finite-step numerical integration), and (ii) the practical effect of the residual-stabilization mechanism of Section~III-D in a controlled off-manifold setting.

Unless otherwise stated, the numerical experiments use the parameter values $(m,g,J_1,J_2,J_3)=(1,9.8,2,2,6)$, with diagonal inertia \(J=\mathrm{diag}(J_1,J_2,J_3)\). The closed-loop dynamics are integrated with a fourth-order Runge--Kutta method. In the results below, we report configuration-level task errors,
\[
e_{\mathrm{pt}}(t)=\|b_1(t)+\hat{p}(t)\|,
\qquad
e_z(t)=|e_3^\top p(t)-z_0|,
\]
together with velocity-level residuals \(\mu_i(t)\), \(i=1,2,3\), defined in \eqref{eq:mu_residuals}, and selected commanded inputs \(f(t),\tau_2(t),\tau_3(t)\), when relevant.

We consider two initialization regimes. In the first (on-manifold) regime, the initial condition satisfies the pointing and altitude constraints together with the corresponding velocity-level compatibility conditions. In the second (structured off-manifold) regime, we perturb only the vertical residual channel \(\mu_3\) while keeping the remaining residuals close to their compatible values. This second regime is used to isolate and visualize the action of the residual-stabilization mechanism in Section~III-D, without masking effects from the more strongly coupled lateral residual dynamics.

In Fig.~\ref{fig:traj_3d_pointing} we show a geometric illustration of the pointing--altitude task. The target is the origin \(p_\star=(0,0,0)\), and the prescribed flight plane is the horizontal plane \(z=z_0\) with \(z_0=0.58\). The displayed trajectory is a representative circular arc on that plane, parameterized with radius \(R_{\mathrm{traj}}=0.9\) and angular span \(\theta\in[20^\circ,150^\circ]\). Along the curve, body-axis snapshots are shown, where the red axis corresponds to \(b_1\), the green axis to \(b_2\), and the blue axis to \(b_3\). The frame is constructed so that \(b_1=-\hat{p}\), with \(\hat{p}=p/\|p\|\), i.e., the red axis is aligned with the line-of-sight from the UAV to the target (black point), while the motion remains on the constant-altitude plane.

\begin{figure}[h!]
    \centering
    \includegraphics[width=0.99\linewidth]{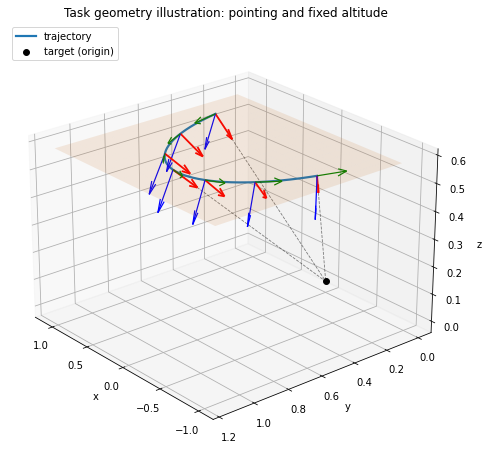}
    \caption{Illustration of the pointing--altitude task geometry. The red body axis \(b_1\) is aligned with the line-of-sight to the target at the origin, while the motion evolves on the prescribed horizontal plane \(z=z_0\).}
    \label{fig:traj_3d_pointing}
\end{figure}

We first simulate the invariance-enforcing law (equivalently, the residual dynamics with \(k_1=k_2=k_3=0\)) from an on-manifold initialization, without numerical projection of the task constraints. Figure~\ref{fig:onmanifold_config_errors} shows the configuration-level task errors \(e_{\mathrm{pt}}(t)\) and \(e_z(t)\). The altitude error remains at numerical zero, whereas a gradual drift appears in the pointing error. This behavior is consistent with numerical drift in the velocity-level compatibility conditions under finite-step integration, and motivates the residual-stabilization extension of Section~III-D for practical simulations.

\begin{figure}[h!]
    \centering
    \includegraphics[width=0.92\linewidth]{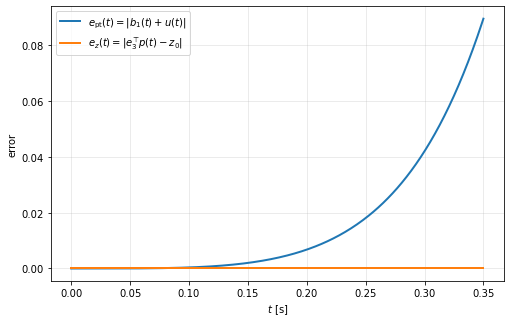}
    \caption{Pure invariance law (\(k_1=k_2=k_3=0\)) with on-manifold initialization and no projection: configuration-level task errors \(e_{\mathrm{pt}}(t)\) and \(e_z(t)\).}
    \label{fig:onmanifold_config_errors}
\end{figure}

To illustrate the effect of residual stabilization in a clean and interpretable way, we next consider a structured off-manifold initialization in which only the linear (vertical) residual is perturbed, i.e., \(\mu_3(0)=e_3^\top v(0)\neq 0\), while \(\mu_1(0)\) and \(\mu_2(0)\) remain close to their compatible values. We compare the pure invariance law (\(k_1=k_2=k_3=0\)) with a linear (vertical) residual-damping law (\(k_1=k_2=0\), \(k_3=5\)) without numerical projection. This experiment is not intended as a full multivariable tuning study of the residual dynamics; rather, it isolates one residual channel to demonstrate the practical effect of the proposed stabilization mechanism in a regularity-preserving regime.

Figure~\ref{fig:vertical_ez_compare} shows the altitude error \(e_z(t)\) for both cases. The vertical residual damping law reduces the altitude drift relative to the pure invariance law. Figure~\ref{fig:vertical_mu3_compare} shows the corresponding evolution of \(\mu_3(t)\): under the pure invariance law, \(\mu_3\) remains approximately constant, whereas with \(k_3=5\) it decays significantly over the simulated interval. The associated thrust input histories in Fig.~\ref{fig:vertical_f_compare} remain bounded and comparable in magnitude, indicating that the improvement is achieved without excessive control effort. In all reported runs, the trajectory remains within the regularity-preserving regime used in the numerical implementation.

\begin{figure}[t]
    \centering
    \includegraphics[width=0.92\linewidth]{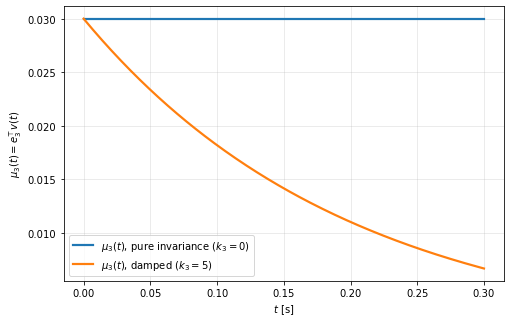}
    \caption{Comparison of the altitude error \(e_z(t)=|e_3^\top p(t)-z_0|\) for a structured off-manifold initialization with \(\mu_3(0)\neq0\): pure invariance law (\(k_3=0\)) versus vertical residual damping (\(k_3=5\), with \(k_1=k_2=0\)). Damping \(\mu_3\) reduces altitude drift over the simulated interval.}
    \label{fig:vertical_ez_compare}
\end{figure}

\begin{figure}[t]
    \centering
    \includegraphics[width=0.92\linewidth]{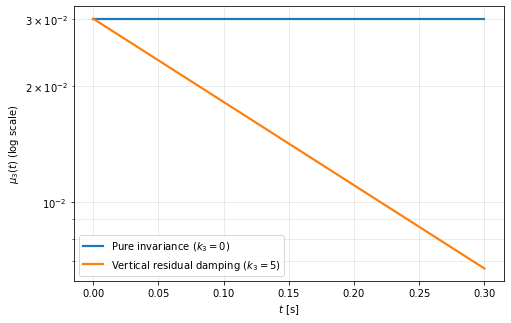}
    \caption{Comparison of the vertical residual \(\mu_3(t)=e_3^\top v(t)\) for the same experiment as in Fig.~\ref{fig:vertical_ez_compare}, shown on a logarithmic scale.}
    \label{fig:vertical_mu3_compare}
\end{figure}

\begin{figure}[t]
    \centering
    \includegraphics[width=0.92\linewidth]{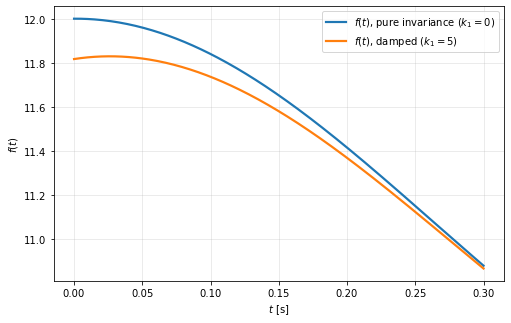}
    \caption{Thrust input \(f(t)\) for the structured vertical-residual experiment. The input remains bounded in both cases and has comparable magnitude, indicating that the reduction in altitude drift is achieved without excessive control effort.}
    \label{fig:vertical_f_compare}
\end{figure}

Finally, to connect the numerical observations with the assumptions underlying the control law, we monitor the regularity indicators that appear in the denominators and in the local validity conditions. Figure~\ref{fig:regularity_monitor_vertical} reports \(|e_3^\top b_3(t)|\) and \(\rho(t)=\|p(t)\|\) for the structured vertical-residual experiment, for both the pure invariance law and the vertically damped case. In both cases, the trajectories remain well separated from the chosen regularity thresholds throughout the simulated interval, which is consistent with the well-posedness of the computed inputs in this local regime.

For completeness, Fig.~\ref{fig:tau_inputs_vertical_case} shows the torque channels \(\tau_2(t)\) and \(\tau_3(t)\) for the vertically damped case (\(k_1=k_2=0\), \(k_3=5\)). The torque inputs remain bounded and smooth, with \(\tau_2\) carrying the dominant actuation effort and \(\tau_3\) remaining comparatively small in magnitude for this experiment. Together with Figs.~\ref{fig:vertical_ez_compare}--\ref{fig:vertical_f_compare}, these results support the practical viability of the proposed invariance/stabilization construction in the regularity-preserving regime considered here.

\begin{figure}[t]
    \centering
    \includegraphics[width=0.92\linewidth]{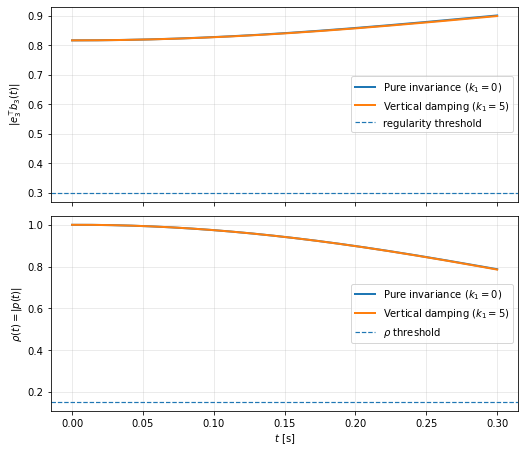}
    \caption{Regularity monitors for the structured vertical-residual experiment (same setup as Figs.~\ref{fig:vertical_ez_compare}--\ref{fig:vertical_f_compare}): top, \(|e_3^\top b_3(t)|\); bottom, \(\rho(t)=\|p(t)\|\). The dashed lines indicate the numerical regularity thresholds used in the simulations.}
    \label{fig:regularity_monitor_vertical}
\end{figure}

\begin{figure}[t]
    \centering
    \includegraphics[width=0.92\linewidth]{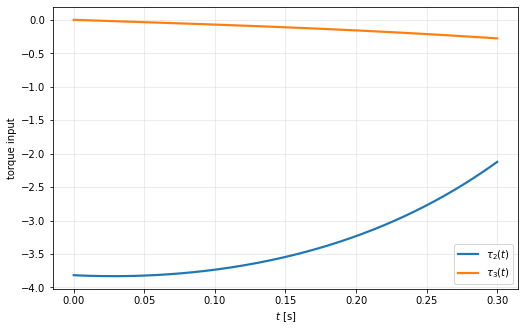}
    \caption{Torque inputs \(\tau_2(t)\) and \(\tau_3(t)\) for the vertically damped structured off-manifold experiment (\(k_1=k_2=0\), \(k_3=5\)). The torque channels remain bounded and smooth over the simulated interval, consistent with operation inside the regularity regime shown in Fig.~\ref{fig:regularity_monitor_vertical}.}
    \label{fig:tau_inputs_vertical_case}
\end{figure}

\section{Conclusions and Future Work}\label{sec:conclusions}

We presented a geometric control framework on $SE(3)$ for quadrotors that achieves pointing-driven missions by construction. The mission is encoded as a virtual nonholonomic constraint distribution (velocity level), whose invariance is enforced by differentiating the constraints along the dynamics and solving the resulting linear system in a selected set of control inputs. Under a transversality condition between the task distribution and the effective actuation distribution, the system is nonsingular and yields a uniquely defined invariance-enforcing feedback law, making feasibility and well-posedness explicit in geometric terms. As a representative application, we considered the combined task of locking a $\tau_1=0$, obtaining closed-form expressions for the thrust and remaining torques that expose the translation–rotation coupling induced by the thrust direction. In addition, we derived a local off-manifold stabilization extension that acts directly on the velocity-level compatibility residuals. On the regularity-preserving set where the denominators are well-defined, this extension enforces exponentially stable first-order residual dynamics and recovers the invariance-enforcing law on the constraint set, providing a practical mechanism to mitigate drift due to initialization mismatch or numerical integration.

Several next steps are worth pursuing. First, we plan to better characterize how the stabilization gains affect transient behavior and robustness in practice. Second, we will apply the same methodology to a broader set of pointing-driven missions, beyond the specific line-of-sight and altitude task considered here. Finally, we aim to validate the approach in experiments on a quadrotor platform to assess performance under sensor noise, disturbances, and actuation limits. Future work also include evaluating the approach in scenarios with additional onboard pointing actuation (e.g., gimbals) to clarify when enforcing vehicle-level geometry is beneficial versus delegating pointing to the payload, and extending the formulation to other underactuated or input-restricted platforms and tasks where not all degrees of freedom are controllable.
\bibliographystyle{IEEEtran}
\bibliography{autosam}



%

\end{document}